\documentclass[12pt]{article}
\usepackage[authoryear]{natbib}
\usepackage{amssymb}
\usepackage{amsfonts}
\usepackage{amsmath}
\usepackage{dsfont}
\usepackage{soul}
\usepackage[nohead]{geometry}
\usepackage{graphicx}
\usepackage{amsthm}
\usepackage{color}
\usepackage{comment}
\usepackage{setspace}
\usepackage{framed}
\usepackage{enumitem}
\usepackage{todonotes}
\usepackage{tikz}
\usepackage{rotating}
\usepackage{subfigure}
\usepackage[flushleft]{threeparttable}
\usepackage{multirow}
\usepackage[colorlinks=true,citecolor=blue,linkcolor=blue]{hyperref}
\usepackage{xr}
\usepackage{bm}
\usepackage{xcolor}
\usepackage{booktabs}
\usepackage{lscape}
\usepackage{algorithm}
\usepackage{algorithmic}
\usepackage[toc,title,page]{appendix}
\usepackage{lscape}
\usepackage{pdflscape}

\usepackage{natbib}

\renewcommand{\today}{\ifcase \month \or January\or February\or March\or %
	April\or May\or June\or July\or August\or September\or October\or November\or %
	December\fi, \number \year} 

\usepackage{caption} 
\allowdisplaybreaks

7

\newcommand{\tcr}{\textcolor{red}}
\newcommand{\eps}[0]{\ensuremath{\varepsilon}}

\newcommand{\bb}{\mbox{\bf b}}

\newcommand{\bx}{\mbox{\bf x}}

\newcommand{\bg}{\mbox{\bf g}}
\newcommand{\bA}{\mbox{\bf A}}
\newcommand{\ba}{\mbox{\bf a}}

\newcommand{\bB}{\mbox{\bf B}}

\newcommand{\bG}{\mbox{\bf G}}
\newcommand{\bH}{\mbox{\bf H}}

\newcommand{\bP}{\mbox{\bf P}}

\newcommand{\bW}{\mbox{\bf W}}

\newcommand{\bV}{\mbox{\bf V}}
\newcommand{\bQ}{\mbox{\bf Q}}
\newcommand{\bR}{\mbox{\bf R}}

\newcommand{\bU}{\mbox{\bf U}}
\newcommand{\bX}{\mbox{\bf X}}

\newcommand{\beps}{\mbox{\boldmath $\varepsilon$}}

\newcommand{\bLambda}{\mbox{\boldmath $\Lambda$}}

\newcommand{\bSigma}{\mbox{\boldmath $\Sigma$}}
\newcommand{\bOmega}{\mbox{\boldmath $\Omega$}}

\newcommand{\bPhi}{\mbox{\boldmath $\Phi$}}
\newcommand{\bPsi}{\mbox{\boldmath $\Psi$}}

\newcommand{\cov}{\mathrm{cov}}

\newcommand{\tr}{\mathrm{tr}}
\newcommand{\diag}{\mathrm{diag}}

\newcommand{\bw}{\mbox{\bf w}}
\newcommand{\var}{\mathrm{var}}
\newcommand{\beq}{\begin{eqnarray*}}
\newcommand{\eeq}{\end{eqnarray*}}

\usepackage{graphicx,rotating,booktabs,threeparttable}

\usepackage{adjustbox}
\usepackage{array}

\newcolumntype{R}[2]{%
	>{\adjustbox{angle=#1,lap=\width-(#2)}\bgroup}%
	l%
	<{\egroup}%
}

\newtheorem{thm}{Theorem}[section]

\newtheorem{lem}{Lemma}[section]

\newtheorem{assum}{Assumption}[section]
\newtheorem{pro}{Proposition}[section]
\numberwithin{equation}{section}
\theoremstyle{definition}

\newtheorem{remark}{Remark}[section]
\makeatletter
\def\@biblabel#1{\hspace*{-\labelsep}}
\makeatother
\DeclareMathOperator*{\argmin}{argmin}

\numberwithin{equation}{section}

\renewcommand{\hat}{\widehat}

\renewcommand{\hat}{\widehat}

\newcommand{\bfm}[1]{\ensuremath{\mathbf{#1}}}

\def\ba{\bfm a}   \def\bA{\bfm A}  
\def\bb{\bfm b}   \def\bB{\bfm B}  
\def\bc{\bfm c}

    \def\FF{\mathbb{F}}
\def\bg{\bfm g}   \def\bG{\bfm G}  
   \def\bH{\bfm H}

   \def\bP{\bfm P}  
   \def\bQ{\bfm Q}  
   \def\bR{\bfm R}

   \def\bU{\bfm U}  
   \def\bV{\bfm V}  
\def\bw{\bfm w}   \def\bW{\bfm W}  
\def\bx{\bfm x}   \def\bX{\bfm X}

 \def\cF{{\cal  F}}

\def\calX{{\cal  X}}

\newcommand{\bfsym}[1]{\ensuremath{\boldsymbol{#1}}}

            \def\bPsi  {\bfsym {\Theta}}
 \def\beps{\bfsym \varepsilon}          
              \def\bSigma{\bfsym \Sigma}
         \def\bLambda {\bfsym {\Lambda}}
           \def\bOmega {\bfsym {\Omega}}

 \def\bPsi{\bfsym {\Psi}}

 \def \bPhi {\bfsym \Phi}

\def\1{\bfsym{1}}	

\def\eps{\varepsilon}
\def\beps{\mbox{\boldmath$\eps$}}
\def\newpage{\vfill\eject}

\def\var{\mbox{var}}

\def\today{\ifcase\month\or
  January\or February\or March\or April\or May\or June\or
  July\or August\or September\or October\or November\or December\fi
  \space\number\day, \number\year}

\newdimen\biblioindent    \biblioindent=30pt

 at 8truept

\def\sign{\mbox{sign}}

\def\eps{\varepsilon}

\newcommand{\beqn}{\begin{eqnarray}}
  \newcommand{\eeqn}{\end{eqnarray}}
\newcommand{\beqnn}{\begin{eqnarray*}}
  \newcommand{\eeqnn}{\end{eqnarray*}}

\allowdisplaybreaks
\usepackage{verbatim}
\def\tilde{\widetilde}

\def\FF{\mathcal{F}}
\def\[{\left [}  \def\]{\right ]} \def\({\left (}  \def\){\right )}
 
\def\hat{\widehat}

\theoremstyle{definition}

\def \diag {\mathrm{diag}} \def \Diag {\mathrm{Diag}}

\begin{document}
	\bibliographystyle{ecta}
	
	\title{Matrix-based Prediction Approach for Intraday Instantaneous Volatility Vector}

	\date{\today}

	\author{
		Sung Hoon Choi\thanks{%
			Department of Economics, University of Connecticut, Storrs, CT 06269, USA. 
			E-mail: \texttt{sung\_hoon.choi@uconn.edu}.} \\ 
		\and Donggyu Kim\thanks{Department of Economics, University of California, Riverside, CA 92521, USA.
			Email: \texttt{donggyu.kim@ucr.edu}.}\\ 
		}
        
	\maketitle
	\pagenumbering{arabic}	
\begin{abstract}
		\onehalfspacing
        In this paper, we introduce a novel method for predicting intraday instantaneous volatility based on Itô semimartingale models using high-frequency financial data. 
        Several studies have highlighted stylized volatility time series features, such as interday auto-regressive dynamics and the intraday U-shaped pattern. 
        To accommodate these volatility features, we propose an interday-by-intraday instantaneous volatility matrix process that can be decomposed into low-rank conditional expected instantaneous volatility and noise matrices. 
        To predict the low-rank conditional expected instantaneous volatility matrix, we propose the Two-sIde Projected-PCA (TIP-PCA) procedure. 
        We establish asymptotic properties of the proposed estimators and conduct a simulation study to assess the finite sample performance of the proposed prediction method. 
        Finally, we apply the TIP-PCA method to an out-of-sample instantaneous volatility vector prediction study using high-frequency data from the S\&P 500 index and 11 sector index funds.
		\\
  
\end{abstract}

\newpage
	
\doublespacing

\section{Introduction}

The analysis of volatility is a vibrant research area in financial econometrics and statistics.
In practice, it is crucial to investigate the volatility dynamics of asset returns for hedging, option pricing, risk management, and portfolio management. 
With the wide availability of high-frequency financial data, several well-performing non-parametric integrated volatility estimation methods have been developed. 
Examples include two-time scale realized volatility (TSRV) \citep{zhang2005tale}, multi-scale realized volatility (MSRV) \citep{zhang2006efficient, zhang2011estimating}, pre-averaging realized volatility (PRV) \citep{christensen2010pre, jacod2009microstructure}, wavelet realized volatility (WRV) \citep{fan2007multi},  kernel realized volatility (KRV) \citep{barndorff2008designing, barndorff2011multivariate},
quasi-maximum likelihood estimator (QMLE) \citep{ait2010high, xiu2010quasi},  local method of moments \citep{bibinger2014estimating},  and robust pre-averaging realized volatility \citep{fan2018robust, shin2023adaptive}.
This incorporation of high-frequency information enhances our understanding of low-frequency (i.e., interday) market dynamics, and several conditional volatility models have been developed based on the realized volatility to explain market dynamics.
Examples include realized volatility-based modeling approaches \citep{andersen2003modeling},
heterogeneous auto-regressive (HAR) models \citep{corsi2009simple}, high-frequency-based volatility (HEAVY) models \citep{shephard2010realising}, realized GARCH models \citep{hansen2012realized}, and unified GARCH-It\^o models \citep{kim2016unified, song2021volatility}.

To understand intraday dynamics, several non-parametric instantaneous (or spot) volatility estimation procedures have been developed \citep{fan2008spot,figueroa2022kernel, foster1996continuous, kristensen2010nonparametric, mancini2015spot, todorov2019nonparametric, todorov2023bias, zu2014estimating}.
With these well-performing instantaneous volatility estimators, several studies have found that intraday instantaneous volatility has U-shaped patterns \citep{admati1988theory, andersen1997intraday, andersen2019time, hong2000trading, li2023robust}.
On the other hand, forecasting intraday volatility has not received as much attention in the literature compared to daily volatility forecasting.
\cite{engle2012forecasting} developed a modified GARCH model that imposes the conditional variance as a product of daily, diurnal, and stochastic intraday components for intraday volatility forecasting.
Recently, \cite{zhang2024volatility} studied several non-parametric machine learning methods to forecast intraday realized volatility by utilizing commonality in intraday volatility.
From the aforementioned studies, it is evident that interday volatility dynamics can be explained by autoregressive-type time series dynamics while intraday volatility dynamics have some periodic patterns, such as a U-shape.
Thus, to predict the one-day-ahead whole intraday instantaneous volatility vector, we need to consider the interday and intraday dynamics simultaneously.
Furthermore, since we consider the prediction of whole intraday instantaneous volatilities, we need to handle the overparameterization issue.

This paper introduces a novel approach for predicting the one-day-ahead instantaneous volatility process.
Specifically, we represent the instantaneous volatility process in a matrix form.
For example, to account for interday time series dynamics and intraday periodic patterns, each row corresponds to a day, and each column represents a high-frequency sequence within that day.
Thus, we have an interday-by-intraday instantaneous volatility matrix. 
To handle the overparameterization issue, we impose a low-rank plus noise structure on the instantaneous volatility matrix, where the low-rank components represent a conditional expected instantaneous volatility matrix and have semiparametric factor structures. 
To accommodate the proposed instantaneous volatility model, we adopt the Projected-PCA method \citep{fan2016projected} to estimate left-singular and right-singular vector components with instantaneous volatility matrix estimators. 
For example, for the left-singular vector, we project the left-singular vectors onto a linear space spanned by past realized volatility estimators, which enables us to account for the interday time-series dynamics and to predict the one-day-ahead instantaneous volatility vector with observed current realized volatility estimators.
Conversely, for the right-singular vector, to explain periodic patterns,   we project the right-singular vectors onto a linear space spanned by deterministic time sequences. 
This two-side projection enables us to explain the interday and intraday dynamics simultaneously. 
 We call this the Two-sIde Projected-PCA (TIP-PCA) procedure. 
We then derive convergence rates for the projected instantaneous volatility matrix estimator and the predicted instantaneous volatility vector using the TIP-PCA method.
As an empirical study, we apply the proposed TIP-PCA estimator on out-of-sample predictions for the one-day-ahead instantaneous volatility process using high-frequency trading data. 
Using a rolling window scheme, we compare TIP-PCA with several methods. 
TIP-PCA demonstrates superior performance by leveraging both interday and intraday dynamics. 
Additionally, we evaluate one-day-ahead 10-minute frequency Value at Risk (VaR) predictions, where TIP-PCA consistently outperforms other methods in backtests.

The remainder of the paper is structured as follows. 
Section \ref{setup} establishes the model and introduces the TIP-PCA prediction procedure. 
Section \ref{asymp} provides an asymptotic analysis of the TIP-PCA estimators. 
The effectiveness of the proposed method is demonstrated through a simulation study in Section \ref{simulation} and by applying it to real high-frequency financial data for predicting the one-day-ahead instantaneous volatility process in Section \ref{empiric}. 
Section \ref{conclusion} concludes the study. 
All proofs are presented in the online supplementary file.

\section{Model Setup and Estimation Procedure} \label{setup}
Throughout this paper, we denote by $\|\bA\|_{F}$, $\|\bA\|_{2}$ (or $\|\bA\|$ for short), $\|\bA\|_{1}$, $\|\bA\|_{\infty}$, and $\|\bA\|_{\max}$ the Frobenius norm,  operator norm, $l_{1}$-norm, $l_{\infty}$-norm, and elementwise norm, which are defined, respectively, as $\|\bA\|_{F} = \tr^{1/2}(\bA'\bA)$, $\|\bA\|_{2} = \lambda_{\max}^{1/2}(\bA'\bA)$, $\|\bA\|_{1} = \max_{j}\sum_{i}|a_{ij}|$, $\|\bA\|_{\infty} = \max_{i}\sum_{j}|a_{ij}|$, and $\|\bA\|_{\max} = \max_{i,j}|a_{ij}|$. 
When $\bA$ is a vector, the maximum norm is denoted as $\|\bA\|_{\infty}=\max_{i}|a_{i}|$, and both $\|\bA\|$ and $\|\bA\|_{F}$ are equal to the Euclidean norm. 
   
\subsection{A Model Setup} \label{model}
We consider the following jump diffusion process: for the $i$-th day and intraday time $t \in [0,1],$
\begin{equation}\label{diffusion-def}
	dZ_{i,t}= \mu_{i,t} dt + \sigma_{i,t}  dB_{i,t} + J_{i,t} d P_{i,t}, 
\end{equation}
where $Z_{i,t}$ is the log price of an asset, $\mu_{i,t}$ is a drift process,  $B_{i,t}$ is a one-dimensional standard Brownian motion,  $J_{i,t}$ is the jump size, and $P_{i,t}$ is the Poisson process with the intensity $\mu_{J}$.
For a given intraday time sequence, for each $i =1,\dots, D$ and $j=1,\dots, n$, we denote the instantaneous volatility process as $c_{i,j} := \sigma_{i,t_{j}}^2$,  where $0< t_{1} < \cdots < t_{n} =1$.
Then, we can write the discrete-time instantaneous volatility process as follows: 
\begin{equation}\label{Sigma}
	\bSigma_{D, n} = (c_{i,j}) _{D \times n}  = \bU \bLambda \bV^{\prime} + \bSigma_{\eps},
\end{equation}
where $\bU=  (u_{i,k} ) _{i=1,\ldots, D, k=1,\ldots,r}$ is the left singular vector matrix, $\bV = (v_{j,k}) _{j=1, \ldots, n, k=1,\ldots, r}$ is the right singular vector matrix, $\bLambda= \Diag (\lambda_1, \ldots, \lambda_r)$ is the singular value matrix, and $\bSigma_{\eps} = (\eps_{i,j})_{D\times n}$ is the random noise matrix.  
We note that $\bSigma_{D, n}$ is a $D \times n$ matrix, which is distinct from a covariance matrix and contains only positive elements.
The left singular vector matrix represents interday volatility dynamics, while the right singular vector matrix explains intraday volatility dynamics.
For example, we consider $r=1$. 
For the intraday, the U-shaped instantaneous volatility pattern is often observed in empirical data and supported by the financial market \citep{admati1988theory, andersen1997intraday, andersen2019time, hong2000trading}. 
Thus, we can use a U-shape function with respect to time $t$ for $\bV$, for example,  $v_{t,1} = a_1 (t - a_2)^2 +a_3$. 
In contrast, for the interday, the daily dynamics are often explained by past realized volatilities \citep{corsi2009simple, hansen2012realized,  kim2019factor, kim2016unified, shephard2010realising, song2021volatility}.
 To reflect this,  $u_{i,1}$ is a function of past realized volatilities, such as the HAR model \citep{corsi2009simple}, for example, $u_{i,1} = b_0 + b_1 RV_{i-1} + b_2  \frac{1}{5}\sum_{j=1}^5 RV_{i-j} +  b_3  \frac{1}{22}\sum_{j=1}^{22} RV_{i-j}$, where $RV_i$ is the $i$-th day realized volatility. 
The features mentioned above motivate the representation of the model (\ref{Sigma}).
We note that in this paper, to simplify prediction modeling, we consider the time series structure of interday volatility dynamics and the instantaneous periodic pattern of intraday volatility dynamics, as described in the example above. 
Our empirical experiment supports that the simplified structure performs well.
However, we can impose more complicated structures on both dimensions.
That is, the above example is one of the possible models, and we propose a generalized model discussed below.

In this paper, our goal is to predict the instantaneous volatility process for the next day.
In general, we assume that $\lambda_1, \ldots, \lambda_r$ are latent singular values,  $u_{i,k}$ is $\mathcal{F}_{i-1}$-adapted, and $v_{t,k}$ is a function of time $t$. 
 Thus,  given $\mathcal{F}_D$, we can predict the instantaneous volatility as follows:
 \begin{equation} \label{conditional exp}
 	E\left[\sigma_{D+1,t}^2 | \FF_{D}\right] =  \sum_{k=1}^r \lambda_k u_{D+1,k} v_{t,k} \text{ a.s.}
 \end{equation}
Given  (\ref{Sigma}), we impose the following nonparametric structure on both singular vectors: for each $k \leq r$, $i \leq D$, and $j \leq n$,
  \begin{align}
  	u_{i,k} = g_{k}(\bx_{i}), \qquad v_{j,k} = h_{k}(\bw_{j}),
  \end{align}
  where $\bx_{i} = (x_{i1},\dots,x_{id_{1}})$ and $\bw_{j} = (w_{j1},\dots,w_{jd_{2}})$ are observable covariates that explain the left and right singular vectors, respectively.
In this context, $\bx_{i}$ can be the past realized volatility of yesterday, last week, and last month in the HAR model, while $\bw_{j}$ can be the intraday time sequence.
Furthermore, we assume that each unknown nonparametric function is additive as follows: for each $k\leq r$, $i\leq D$, and $j\leq n$,
 \begin{align*}
  	  &g_{k}(\bx_{i}) = \phi(\bx_{i})'\bb_{k} + R_{k}(\bx_{i}),\\
      &h_{k}(\bw_{j}) = \psi(\bw_{j})'\ba_{k} + Q_{k}(\bw_{j}),
  \end{align*}
  where $\phi(\bx_{i})$ is a $(J_{1}d_{1})\times 1$ vector of basis functions, $\bb_{k}$ is a $(J_{1}d_{1})\times 1$ vector of sieve coefficients, and $R_{k}(\bx_{i})$ is the approximation error term; $\psi(\bw_{j})$ is a $(J_{2}d_{2})\times 1$ vector of basis functions, $\ba_{k}$ is a $(J_{2}d_{2})\times 1$ vector of sieve coefficients, and $Q_{k}(\bw_{j})$ is the approximation error term.
  Equivalently, we can write $g_{k}(\bx_{i}) = \sum_{d=1}^{d_{1}}g_{kd}(x_{id})$ and $h_{k}(\bw_{j}) = \sum_{d=1}^{d_{2}}h_{kd}(w_{jd})$, where $g_{kd}(x_{id}) = \sum_{l=1}^{J_{1}}b_{l,kd}\phi_{l}(x_{id})+R_{kd}(x_{id})$ and $h_{kd}(w_{jd}) = \sum_{l=1}^{J_{2}}a_{l,kd}\psi_{l}(w_{id})+Q_{kd}(w_{jd})$, respectively.
  Hence, each additive component of $g_{k}$ and $h_{k}$ can be estimated by the sieve method.
  Throughout the paper, we assume that $d_{1}=\dim(\bx_{i})$, $d_{2} = \dim(\bw_{j})$ and $r$ are fixed.
  The number of sieve terms, $J_{1}$ and $J_{2}$, grow very slowly as $D \rightarrow \infty$ and $n \rightarrow \infty$, respectively.
    In a matrix form, we can write
  \begin{align*}
  	&\bU := \bG(\bX) = \bPhi(\bX)\bB + \bR(\bX),\\
  	&\bV := \bH(\bW) = \bPsi(\bW)\bA + \bQ(\bW),
  \end{align*}
  where the $D \times (J_{1}d_{1})$ matrix $\bPhi(\bX) = (\phi(\bx_{1}), \dots, \phi(\bx_{D}))'$, the $(J_{1}d_{1})\times r$ matrix $\bB = (\bb_{1},\dots, \bb_{r})$, and $\bR(\bX) = (R_{k}(\bx_{i}))_{D\times r}$; the $n \times (J_{2}d_{2})$ matrix $\bPsi(\bW) = (\psi(\bw_{1}), \dots, \psi(\bw_{n}))'$, the $(J_{2}d_{2})\times r$ matrix $\bA = (\ba_{1},\dots, \ba_{r})$, and $\bQ(\bW) = (Q_{k}(\bw_{j}))_{n\times r}$.
 We note that since $\bG(\bX)$ and $\bH(\bW)$ are identifiable up to the sign and $\bPhi(\bX)$ and $ \mathbf{\Psi} (\bW)$ are given, the coefficients $\bA$ and $\bB$ can be defined by the definition of linear regression coefficients.
    Then, since the vectors of  basis functions are non-singular, we can define uniquely $\bPhi(\bX)\bB$ and $ \mathbf{\Psi}(\bW)\bA$.

 Due to the imperfections of the trading mechanisms \citep{ait2009high}, the true underlined log-stock price $Z_{i,t}$ in \eqref{diffusion-def} is not observable. 
 To reflect the imperfections, we assume that the high-frequency intraday observations $X_{i,t_s}, s=1, \ldots, m,$ are contaminated by microstructure noises as follows:
	\begin{equation} 	\label{def-obervation}
	Y _{i,t_s} = Z_{i,t_s} + e_{i,t_s}, \quad i=1,\ldots,D, s=1,\ldots,m,
	\end{equation}
where the microstructure noises $e_{i,t_s}$ are independent random variables with a mean of zero and a variance of $\eta_{ii}$. 
For simplicity, we assume that the observed time points are equally spaced, that is, 
$t_{s}- t_{s-1} = m ^{-1}$ for $i=1,\ldots, D$ and $s=2,\ldots, m$.

Several non-parametric instantaneous volatility estimation procedures have been developed \citep{fan2008spot,figueroa2022kernel, foster1996continuous, kristensen2010nonparametric, mancini2015spot, todorov2019nonparametric, todorov2023bias, zu2014estimating}.
We can use any well-performing instantaneous volatility estimator that satisfies Assumption \ref{assum_spotvol} (ii). 
In the numerical study, we employ the jump robust pre-averaging method proposed by \cite{figueroa2022kernel}. 
The specific method is described in \eqref{pre-spot}.

\subsection{Two-sIde Projected-PCA} \label{estimation procedure}

To accommodate the semiparametric structure in Section \ref{model}, we need to project the left and right singular vectors on linear spaces spanned by the corresponding covariates.
To do this, we apply the Projected-PCA \citep{fan2016projected} procedure to the left and right singular vectors with the well-performing instantaneous volatility estimator. 
The specific procedure is as follows:
\begin{enumerate}
  	\item For each $i \leq D$ and $j \leq n$, we estimate the instantaneous volatility, $c_{i,j} = \sigma_{i,t_{j}}^2$, using high-frequency log-price observations and denote them $\hat{c}_{i,j}$. 
    Let $\hat{\bLambda} = \diag(\hat{\lambda}_{1}, \dots, \hat{\lambda}_{r}),$ where $\hat{\lambda}_{1}\geq \hat{\lambda}_{2}\geq \cdots \geq \hat{\lambda}_{r}$ are the square root of the leading eigenvalues of $\hat{\bSigma}_{D,n}\hat{\bSigma}_{D,n}^{\prime}$, where $\hat{\bSigma}_{D,n} = (\hat{c}_{i,j})_{D \times n}$.
    
  	\item Define the $D\times D$ projection matrix as $\bP_{\Phi} = \bPhi(\bX)(\bPhi(\bX)'\bPhi(\bX))^{-1}\bPhi(\bX)'$.
  	The columns of $\hat{\bG}(\bX):= (\hat{U}_{1},\dots, \hat{U}_{r})$ are defined as the $r$ leading eigenvectors of the $D\times D$ matrix $\bP_{\Phi}\hat{\bSigma}_{D,n}\hat{\bSigma}_{D,n}^{\prime}\bP_{\Phi}$.
  	Then, we can estimate $\bB$ by
  	\begin{align*}
  	\hat{\bB} = (\hat{\bb}_{1},\dots,\hat{\bb}_{r}) = (\bPhi(\bX)'\bPhi(\bX))^{-1}\bPhi(\bX)'\hat{\bG}(\bX).
  	\end{align*}
   Given any $\bx \in \calX$, we estimate $g_{k}(\cdot)$ by
   $$
   \hat{g}_{k}(\bx) = \phi(\bx)'\hat{\bb}_{k}  \qquad \text{for} \quad k = 1,\dots, r,
   $$
   where $\calX$ denotes the support of $\bx_{i}$.
   
   \item Define the $n\times n$ projection matrix as $\bP_{\Psi} = \bPsi(\bW)(\bPsi(\bW)'\bPsi(\bW))^{-1}\bPsi(\bW)'$.
  	The columns of $\ddot{\bH}(\bW) := (\hat{V}_{1},\dots, \hat{V}_{r})$ are defined as the $r$ leading eigenvectors of the $n\times n$ matrix $\bP_{\Psi}\hat{\bSigma}_{D,n}^{\prime}\hat{\bSigma}_{D,n}\bP_{\Psi}$.
    
    \item We estimate a sign vector $s_{0} = (s_{01},\dots,s_{0r})\in \{-1,1\}^{r}$ defined in \eqref{s0_def} by
    \begin{equation}\label{sign}
    \hat{s} := (\hat{s}_{1},\dots,\hat{s}_{r}) = \argmin_{s\in\{-1,1\}^{r}}\left\|\sum_{k=1}^{r}s_{k}\hat{\lambda}_{k}\hat{U}_{k}\hat{V}_{k}^{\prime}-\hat{\bSigma}_{D,n}\right\|_{F}^{2}.
   \end{equation}
    Then, we update the right singular vector estimator by
    $\hat{\bH}(\bW) = (\hat{s}_{1}\hat{V}_{1},\dots, \hat{s}_{r}\hat{V}_{r})$.
   
   \item Finally, we predict the conditional expectation of the one-day-ahead instantaneous volatility vector $E[(c_{D+1,1},\dots,c_{D+1,n}) | \FF_{D}]$ by
   $$
   (\tilde{c}_{D+1,1},\dots, \tilde{c}_{D+1,n}) = \hat{\bg}(\bx_{D+1})\hat{\bLambda}\hat{\bH}(\bW)^{\prime},
   $$
   where $\hat{\bg}(\bx_{D+1}) = (\hat{g}_{1}(\bx_{D+1}),\dots, \hat{g}_{r}(\bx_{D+1}))$.
   
  \end{enumerate}

\begin{remark}
    To estimate the instantaneous volatility matrix, we need to match the signs for singular vector estimators \citep{cho2017asymptotic}.
    This is because $\hat{U}_{k}$ and $\hat{V}_{k}$ can estimate $U_{k} = (u_{1,k},\dots,u_{D,k})'$ and $V_{k} = (v_{1,k},\dots,v_{D,k})'$ up to signs, such as $\sign(\langle\hat{U}_{k},U_{k}\rangle)$ and $\sign(\langle\hat{V}_{k},V_{k}\rangle)$, respectively.
    Let $s_{0} =(s_{01},\dots,s_{0r})\in \{-1,1\}^{r}$ be
    \begin{equation} \label{s0_def}
    s_{0k} = \sign(\langle\hat{U}_{k},U_{k}\rangle) \sign(\langle\hat{V}_{k},V_{k}\rangle)  \qquad \text{for} \quad k = 1,\dots, r.
    \end{equation}
    Then, $\bU\bLambda\bV^{\prime}$ can be consistently estimated by $\sum_{k=1}^{r}s_{0k}\hat{\lambda}_{k}\hat{U}_{k}\hat{V}_{k}^{\prime}$.
    However, since $s_{0}$ is unknown in practice, we employ the sign estimation procedure as discussed in Step \ref{sign} above.
    We can show their sign consistency under a regularity condition (see Assumption \ref{assum_sign}).
\end{remark}
 
In summary, given the estimated instantaneous volatility matrix, we initially estimate the singular values using the conventional PCA method.
We employ the Projected-PCA method \citep{fan2016projected} to estimate the unknown nonparametric function using observable covariates (e.g., a series of past realized volatilities). 
We then apply the Projected-PCA method again to estimate the right singular vector matrix using the covariate (e.g., intraday time sequence).
Finally, with the observable covariates $\bx_{D+1}$ (i.e., information about past realized volatilities on the $D$th day), we predict the one-day-ahead instantaneous volatility process by multiplying the estimated singular value and vector components.
We refer to this procedure as the Two-sIde Projected-PCA (TIP-PCA). 
The TIP-PCA method can accurately predict instantaneous volatility by incorporating both interday and intraday dynamics, as the projection approach removes noise components.
Moreover, since we assume a low-rank matrix, we can reduce the complexity of the model, which helps overcome the overparameterization.  
The numerical study in Sections \ref{simulation} and \ref{empiric} demonstrates that TIP-PCA performs well in predicting a one-day-ahead instantaneous volatility vector.

\subsection{Choice of Tuning Parameters} \label{tuning}

The suggested TIP-PCA estimator requires the choice of tuning parameters $r$, $J_{1}$, and $J_{2}$. 
First, the number of latent factors, $r$, can be chosen through data-driven methods \citep{ahn2013eigenvalue, bai2002determining, onatski2010determining}.
For example, $r$ can be determined by finding the largest singular value gap or singular value ratio such that $\max_{k\leq r_{\max}}(\hat{\lambda}_{k}-\hat{\lambda}_{k+1})$ and $\max_{k\leq r_{\max}}\frac{\hat{\lambda}_{k}}{\hat{\lambda}_{k+1}}$ for a predetermined maximum number of factors $r_{\max}$.
For the numerical studies in Sections \ref{simulation} and \ref{empiric}, we employed rank 1 using the eigenvalue ratio method proposed by \cite{ahn2013eigenvalue}.

The numbers of sieve terms, $J_{1}$ and $J_{2}$, and basis functions can be flexibly chosen by practitioners based on the conjecture of the nonparametric function form \citep{chen2020semiparametric, fan2016projected}.
We note that $J_{1}$ and $J_{2}$ are the cost to approximate the unknown nonparametric functions $g_{k}(\cdot)$ and $h_{k}(\cdot)$ (see Remark \ref{remark1}).
Both $J_{1}$ and $J_{2}$ grow very slowly such as $\log D$ and $\log n$, respectively, and we require $J_{1}=o(\sqrt{D})$ and $J_{2} = o(\sqrt{n})$ as described in Proposition \ref{element wise norm} and Theorem \ref{main_thm}.
In this context, the interday volatility dynamic is a linear function of the past realized volatilities, while the intraday volatility dynamic is a U-shaped function with respect to deterministic time sequences. 
Therefore, for the numerical studies in Sections \ref{simulation} and \ref{empiric}, we employed the additive polynomial basis with the sieve dimensions $J_{1}=2$ and $J_{2}=3$ for the TIP-PCA method.

\section{Asymptotic Properties} \label{asymp}
In this section, we establish the asymptotic properties of the proposed TIP-PCA estimator.
To do this, we impose the following technical assumptions.

  \begin{assum} \label{assum_spotvol} ~
      \begin{itemize}
            \item[(i)] For $k \leq r$, the eigengap satisfies $|\lambda_{k+1} - \lambda_{k}| = O_{P}(\sqrt{nD})$ and $\lambda_{r+1} = 0$.
            \item[(ii)] For some bounded $\mu_{1}$ and $\mu_{2}$, the left and right singular vectors satisfy: 
            $$
            \max_{i\leq D, k\leq r}|u_{i,k}|^2 \leq \frac{r\mu_1}{D}, \text{ and } \max_{j\leq n, k\leq r}|v_{j,k}|^2 \leq \frac{r\mu_2}{n}.
            $$
            \item[(iii)] 
              
            For each $i\leq D$ and $j\leq n$ and, the estimated instantaneous volatility $\hat{c}_{i,j}$ satisfies
            $$
             \hat{c}_{i,j}-c_{i,j}  = \upsilon_{i,j} + \varsigma_{i,j},
            $$
            where $\upsilon_{i,j}$ follows the martingale difference sequence and $\varsigma_{i,j}$ is the estimation bias term such that $E(\upsilon_{i,j} | \cF_{i,t_j})=0$ a.s.,  $E(\upsilon_{i,j}^{6}) = O(m^{-3/4})$,   $E(\varsigma_{i,j}^{6}) = O(\rho_{m}^{6})$, and for any $s \leq D$ and $j'\leq n$,  $E(  c_{i,j'}\upsilon_{i,j} | \cF_{i,t_j}) = 0$ and $E( c_{s,j} c_{s,j'}\upsilon_{i,j} | \cF_{i,t_j}) = 0$ a.s. Additionally, there exist a constant $C_{1}>0$ such that $E(c_{i,j}^{4}) < C_{1}$.

        \end{itemize}
  \end{assum}
  
Assumption \ref{assum_spotvol} is related to assumptions for the instantaneous volatility matrix.
Assumption \ref{assum_spotvol}(i) is the eigengap assumption, which is essential for analyzing low-rank matrices \citep{candes2010matrix, cho2017asymptotic, fan2018large}.
We note that since we have a $D\times n$ instantaneous volatility matrix, the pervasive condition implies that the eigenvalue for the low-rank component has $\sqrt{nD}$ order. 
Assumption \ref{assum_spotvol}(ii) is the conventional incoherence condition, which is essential for controlling entrywise deviations in low-rank matrix estimation.
The moment conditions in Assumption \ref{assum_spotvol}(iii) can be satisfied under some assumptions on the process $X$, microstructure noise, and kernel function in \cite{figueroa2022kernel} by using similar arguments in \citet{kim2016asymptotic}.
When estimating the spot volatility at time $t$, we usually use data after time $t$.
This implies $E(\upsilon_{i,j} | \cF_{i,t_j})=0$ a.s. 
We note that the spot volatility estimator is asymptotically unbiased, thus, the estimation bias term goes to zero faster than $m^{-\frac{1}{8}}$.
 The uncorrelated condition, $E( c_{i,j'}\upsilon_{i,j} | \cF_{i,t_j})= E( c_{s,j} c_{s,j'}\upsilon_{i,j} | \cF_{i,t_j}) = 0$ a.s., is required to show the benefit of the smoothing scheme. 
This condition is satisfied, if the spot volatility can be represented by  $\cF_{i,t_j}$-adapted processes and independent noises.
For example, the usual time series structure, such as ARMA, satisfies it.
Thus, this condition is not restrictive. 
It is worth noting that if the uncorrelated condition is not satisfied, we have a slower convergence rate, for example, we have additionally $m^{-1/8}$ in Proposition \ref{element wise norm} and Theorem \ref{main_thm}. 

    \begin{assum}\label{assum_phi} ~
        \begin{itemize}
              \item[(i)] There are $c_{\min}$ and $c_{\max}>0$ so that, with the probability approaching one, as $D \rightarrow \infty$ and $n \rightarrow \infty$,
              \begin{align*}
              &c_{\min} < \lambda_{\min}(D^{-1}\bPhi(\bX)'\bPhi(\bX))< \lambda_{\max}(D^{-1}\bPhi(\bX)'\bPhi(\bX))<c_{\max},\\          
              &c_{\min} < \lambda_{\min}(n^{-1}\bPsi(\bW)'\bPsi(\bW))< \lambda_{\max}(n^{-1}\bPsi(\bW)'\bPsi(\bW))<c_{\max}.                    
              \end{align*}          
              \item[(ii)] $\max_{l\leq J_{1},i\leq D, d\leq d_{1}} E\phi_{l}(x_{id})^{2} < \infty$, and $\max_{l\leq J_{2},j\leq n, d\leq d_{2}} E\psi_{l}(w_{jd})^{2} < \infty$.
        \end{itemize}
    \end{assum}
Assumption \ref{assum_phi} is related to basis functions. 
Intuitively, the strong law of large numbers implies Assumption \ref{assum_phi}(i), which can be satisfied by normalizing common basis functions such as B-splines, polynomial series, or Fourier basis.

  \begin{assum} \label{assum_sieve}
      For all $d\leq d_{1}, d'\leq d_{2}, k \leq r$, 
      \begin{itemize}
          \item[(i)] the functions $g_{kd}(\cdot)$ and $h_{kd'}(\cdot)$ belong to a H\"older class $\mathcal{G}$ and $\mathcal{H}$ defined by, for some $L>0$,
		\begin{align*}
			&\mathcal{G} = \{g : |g^{(c)}(s)-g^{(c)}(t)| \leq L|s-t|^{\alpha}\},\\
   			&\mathcal{H} = \{h : |h^{(c')}(s)-h^{(c')}(t)| \leq L|s-t|^{\alpha'}\}.
        \end{align*}
          \item[(ii)] the sieve coefficients $\{b_{l,kd}\}_{l\leq J_{1}}$ satisfy for $\kappa = 2(c+\alpha) \geq 4$, as $J_{1} \rightarrow \infty$, 
		\begin{equation*}
			\sup_{x \in \mathcal{X}_{d}}\left|g_{kd}(x)-\sum_{l=1}^{J_{1}}b_{l,kd}\phi_{l}(x)\right|^{2} = O(J_{1}^{-\kappa}/D),
		\end{equation*}
		where $\mathcal{X}_{d}$ is the support of the $d$th element of $\bx_{i}$.
        Similarly, the sieve coefficients $\{a_{l,kd'}\}_{l\leq J_{2}}$ satisfy for $\kappa = 2(c'+\alpha') \geq 4$, as $J_{2} \rightarrow \infty$, 
		\begin{equation*}
			\sup_{w \in \mathcal{W}_{d'}}\left|h_{kd'}(w)-\sum_{l=1}^{J_{2}}a_{l,kd'}\psi_{l}(w)\right|^{2} = O(J_{2}^{-\kappa}/n),
		\end{equation*}
		where $\mathcal{W}_{d'}$ is the support of the $d'$th element of $\bw_{j}$.
          \item[(iii)] $\max_{k\leq r,l\leq J_{1},d\leq d_{1}}b_{l,kd}^2 = O(1/D)$ and $\max_{k\leq r,l\leq J_{2}, d' \leq d_{2}} a_{l,kd'}^2 = O(1/n)$.
        \end{itemize}
  \end{assum}
Assumption \ref{assum_sieve} pertains to the accuracy of the sieve approximation and can be satisfied using a common basis such as a polynomial basis or B-splines (see \citealp{chen2007large}). 
We have the following conditions that the idiosyncratic errors are weakly dependent on both dimensions, which are commonly imposed for high-dimensional factor analysis.

  \begin{assum} \label{DGP} ~
      \begin{itemize}
        \item[(i)] $E\eps_{i,j} = 0$ for all $i\leq D, j\leq n$; $\{\eps_{i,j}\}_{i\leq D, j\leq n}$ is independent of $\{\bx_{i}, \bw_{j}\}_{i\leq D, j \leq n}$.
          \item[(ii)] There is $C>0$ such that
          \begin{align*}
            \max_{m\leq D}\sum_{i=1}^{D}|E\eps_{i,j}\eps_{m,j}| < C,\\
            \frac{1}{nD}\sum_{i=1}^{D}\sum_{m=1}^{D}\sum_{j=1}^{n}\sum_{s=1}^{n}|E\eps_{i,j}\eps_{m,s}| <  C.
          \end{align*}
          \item[(iii)]  The covariance matrix $\bOmega_{\eps, j} =\cov(\beps_{j}) := (\omega_{\eps,ij})_{D\times D}$ for each $j = 1,\dots, n $, where $\beps_{j} = (\eps_{1,j},\ldots, \eps_{D,j})'$, satisfies
          $$
          \max_{j\leq n}\|\bOmega_{\eps, j }\| = O(\varphi_D), \text{ where } \varphi_D = o(D).
          $$
        \item[(iv)] There are $C_{1}>0$ and $C_{2}>0$ such that, for each $j, l\leq n$ and each $k\leq r$, 
        \begin{align*}            &|\cov(\sum_{i=1}^{D}u_{i,k}\eps_{i,j},\sum_{i=1}^{D}u_{i,k}\eps_{i,l})|\leq C_{1}\varphi_{D}\rho(|j-l|), \text{ where } \sum_{h=1}^{\infty}\rho(h) <\infty, \\
            & |\cov((\sum_{i=1}^{D}u_{i,k}\eps_{i,j})^2,(\sum_{i=1}^{D}u_{i,k}\eps_{i,l})^2)|\leq C_{2}\varphi_{D}^2\tilde{\rho}(|j-l|), \text{ where } \sum_{h=1}^{\infty}\tilde{\rho}(h) < \infty.
        \end{align*}
         In addition, for each $j \leq  n$ and $k \leq r$,         $E[(\sum_{i=1}^{D}u_{i,k}\eps_{i,j})^4] = O(\varphi_{D}^2).$
      \end{itemize}
  \end{assum}
Assumption \ref{DGP}(iii) is the conventional condition on the idiosyncratic covariance matrix, which allows for heterogeneity across the intraday dimension. 
For example, this includes the sparsity condition, which has been considered in many applications \citep{boivin2006more, fan2016incorporating}.
Assumption \ref{DGP}(iv) imposes weak temporal dependence on both the linear and squared projected error terms, and a fourth moment condition.

    \begin{assum}\label{assum_sign}
        The estimated instantaneous volatility matrix $\hat{\bSigma}_{D,n}$ and initial estimators $\{\hat{\lambda}_{k},\hat{U}_{k},\hat{V}_{k}\}_{k=1}^{r}$ satisfy
            $$\lim_{D\rightarrow \infty, n\rightarrow \infty}\mathbb{P}\left(\min_{s\in \{-1,1\}^{r}}\left\|\sum_{k=1}^{r}s_{k}\hat{\lambda}_{k}\hat{U}_{k}\hat{V}_{k}^{\prime}-\hat{\bSigma}_{D,n}\right\|_{F}^{2} < \left\|\sum_{k=1}^{r}s_{0k}\hat{\lambda}_{k}\hat{U}_{k}\hat{V}_{k}^{\prime}-\hat{\bSigma}_{D,n}\right\|_{F}^{2}\right) = 0.$$
    \end{assum}
\begin{remark}\label{sign_remark}
Assumption \ref{assum_sign} is related to the sign estimation in  \eqref{sign}.
In this paper, we conduct the singular value decomposition to estimate the left and right singular vectors separately, but these vectors can only be estimated up to a sign due to the sign problem of singular vectors.
 Hence, to define the signs uniquely, we impose this identifiability condition.
To understand Assumption \ref{assum_sign}, for simplicity, we consider the case of $r = 1$.
When $s_0 = 1$, Assumption \ref{assum_sign} implies
$$
\lim_{D\rightarrow \infty, n\rightarrow \infty}\mathbb{P}\left(\|-\hat{\lambda}_{1}\hat{U}_{1}\hat{V}_{1}^{\prime}-\hat{\bSigma}_{D,n}\|_{F}^{2} < \|\hat{\lambda}_{1}\hat{U}_{1}\hat{V}_{1}^{\prime}-\hat{\bSigma}_{D,n}\|_{F}^{2}\right) = 0.
$$
That is, the probability that $\hat{s}$ chooses a different sign than the true sign goes to zero as dimensions increase \citep{cho2017asymptotic}.
Thus, Assumption \ref{assum_sign} guarantees the identifiability of the sign problem. 
In light of this, Assumption \ref{assum_sign} is the natural assumption to make.
\end{remark}

We obtain the following elementwise convergence rate of the projected instantaneous volatility matrix estimator.
    \begin{pro} \label{element wise norm}
    Suppose that Assumptions \ref{assum_spotvol}--\ref{assum_sign} hold, $J_{1} = o(\sqrt{D})$, and $J_{2} = o(\sqrt{n})$. As $D,n,m,J_{1},J_{2} \rightarrow \infty$, we have
        \begin{align*}
            &\|\hat{\bG}(\bX)\hat{\bLambda}\hat{\bH}(\bW)^{\prime} - \bG(\bX)\bLambda\bH(\bW)^{\prime}\|_{\max}\\ 
            & \qquad  =O_{P}\Bigg(\frac{\varphi_D}{D}+ \max(J_{1},J_{2})\left(\sqrt{\frac{\rho_{m}}{m^{\frac{1}{8}}}} + \rho_{m}+  \frac{1}{\sqrt{\min (n m^{\frac{1}{4}}, D m^{\frac{1}{4}}, nD)}} \right) \\
            &\qquad\qquad\qquad\qquad  + \min(J_{1},J_{2})^{\frac{1}{2}-\frac{\kappa}{2}} + \frac{J_{1}^2}{D} + \frac{J_{2}^2}{n}\Bigg).  
        \end{align*}
    \end{pro}

\begin{remark} \label{remark1}
Proposition \ref{element wise norm} shows that the projected instantaneous volatility matrix estimator has the convergence rate $J\sqrt{\frac{\rho_m}{m^{1/8}}} + J^{\frac{1}{2}-\frac{\kappa}{2}} + \frac{J}{\sqrt{\min (n m^{1/4}, D m^{1/4}, nD)}} +\frac{J^2}{\min (D, n)}$ up to the sparsity level, when $\rho_m < m^{-1/8}$ and $J =J_{1}=J_{2}$.
We note that the convergence rate $\rho_m$ of the bias term is faster than the convergence rate  $m^{-1/8}$ of the martingale difference term in the spot volatility estimator, and the term $\sqrt{\frac{\rho_m}{m^{1/8}}}$ arises from the cross product between two components.
The bias comes from the drift term that has the order of $m^{-1}$, and due to the subsampling scheme to handle the microstructure noise, $\rho_m$ usually is the order of $m^{-1/2}$. 
The martingale term is negligible because of the smoothing effect when calculating the singular components.
It is worth noting that if $\rho_m$ is zero,  with the fixed $m$, we can obtain the consistency.
The terms $\{\frac{1}{\sqrt{nD}}, \frac{1}{D}\}$ represents the cost of learning daily time series dynamics, while the terms $\{\frac{1}{\sqrt{nD}}, \frac{1}{n}\}$ corresponds to the cost of learning intraday periodic patterns via the right singular vector estimator.
We note that $J_1$ and $J_2$ are related to the sieve approximation. 
That is, $\max (J_1, J_2)$ is the cost to approximate the unknown nonparametric functions $g_k(\cdot)$ and $h_k(\cdot)$.  
In addition, the term $\min(J_{1},J_{2})^{\frac{1}{2}-\frac{\kappa}{2}}$ is due to the approximation error.
If $g_k(\cdot)$ and $h_k(\cdot)$ are known, the convergence rate is $\sqrt{\frac{\rho_m}{m^{1/8}}} + \min (n m^{1/4}, D m^{1/4}, nD,n^2, D^2)^{-1/2}$ up to the sparsity level with the finite $J_1$ and $J_2$.
\end{remark}

The following theorem provides the convergence rate of the predicted instantaneous volatility using the TIP-PCA method.
    \begin{thm} \label{main_thm}
    Suppose that Assumptions \ref{assum_spotvol}--\ref{assum_sign} hold, $J_{1} = o(\sqrt{D})$, and $J_{2} = o(\sqrt{n})$. As $D,n,m,J_{1},J_{2} \rightarrow \infty$, we have
        \begin{align*}
            &\max_{j\leq n}\left|\tilde{c}_{D+1,j} - E[c_{D+1,j} | \FF_{D}]\right| \\
            & \qquad = O_{P}\left(\frac{\varphi_D}{D}  + J_{2}\left(\sqrt{\frac{\rho_{m}}{m^{\frac{1}{8}}}} + \rho_{m} + \frac{1}{\sqrt{\min (n m^{\frac{1}{4}}, D m^{\frac{1}{4}}, nD)}}\right) + J_{2}^{\frac{1}{2}-\frac{\kappa}{2}} + \frac{J_{2}^2}{n}\right) \nonumber\\
            &\qquad\qquad + O_{P}\left( J_{1}^{\frac{3}{2}}\left(\sqrt{\frac{\rho_{m}}{m^{\frac{1}{8}}}}+\rho_{m}+ \frac{1}{\sqrt{\min (n m^{\frac{1}{4}}, D m^{\frac{1}{4}}, nD)}}\right)+ J_{1}^{1 - \frac{\kappa}{2}}+ \frac{J_{1}^{\frac{5}{2}}}{D} \right)\max_{l\leq J_{1}}\sup_{x}|\phi_{l}(x)|.
        \end{align*}    
    \end{thm}
Theorem \ref{main_thm} indicates that the proposed TIP-PCA consistently predicts the one-day-ahead instantaneous volatility.
As discussed in Remark \ref{remark1}, we have $\sqrt{\frac{\rho_m}{m^{1/8}}}$, $D^{-1}$, $n^{-1}$, $\min (n m^{1/4}, D m^{1/4}, nD)^{-1/2}$ and $J_i$'s terms.
To predict the one-day-ahead instantaneous volatility process, we need to learn the interday time series dynamics and intraday periodic patterns. 
Thus, the terms $D^{-1}$ and $n^{-1}$ are the costs based on the projection-based estimation approach.
 We note that the difference from the result in Proposition \ref{element wise norm} arises from the estimation of the nonparametric function $g_{k}(\cdot)$.
Specifically, out-of-sample predictions with any covariates $\bx$ necessitate additional costs such as $J_{1}^{1/2}$ and the supremum term $\max_{l\leq J_{1}}\sup_{x}|\phi_{l}(x)|$, whereas the result in Proposition \ref{element wise norm} does not require them because it is based on in-sample prediction.

\section{Simulation Study} \label{simulation}
In this section, we conducted simulations to examine the finite sample performance of the proposed TIP-PCA method.
We first generated high-frequency observations as follows: for $i= 1,\dots, D+1$, $s=0,\dots, m$, and $t_{s} = s/m$,
    \begin{align*}
        &Y_{i,t_{s}} = Z_{i,t_{s}} + e_{i,t_{s}},\\
        &dZ_{i,t} = (\mu - \sigma_{i,t}^2/2)dt + \sigma_{i,t} dB_{i,t} + J_{i,t}dP_{i,t},\\
        &\sigma_{i,t_{s}}^2 = \tilde{\sigma}_{i}^2 h(t_{s}) + \varepsilon_{i,t_{s}},
    \end{align*}
    where we set microstructure noise as $e_{i,t_{s}} \sim \mathcal{N}(0,0.0005^2)$ and the initial value as $Z_{1,0} = 1$; $B_{i,t}$ is a standard Brownian motion; for the jump part, we set $J_{i,t} \sim \mathcal{N}(-0.01, 0.02^2)$ and $P_{i,t+\Delta} - P_{i,t} \sim \text{Poisson}(36\Delta/252)$; $\tilde{\sigma}_{i} = b_{0} + b_{1}\tilde{\sigma}_{i-1} + b_{2}\frac{1}{5}\sum_{s=1}^{5}\tilde{\sigma}_{i-s} + b_{3}\frac{1}{22}\sum_{s=1}^{22}\tilde{\sigma}_{i-s} + \zeta_{i}$,  $\zeta_{i} \sim \mathcal{N}(0,1)$,  $h(t_{s}) = \gamma_{0} + \gamma_{1}(t_{s}-0.6)^2$, and $\varepsilon_{i,t_{s}} = q(t_{s})\xi_{i,t_{s}}$, where $q(t_{s})^{2} = 0.1+0.5(2t_{s}-1)^{2}$ and $\xi_{i,t_{s}} \sim \mathcal{N}(0,0.01^2)$.
The model parameters were set to be
    \begin{align*}
    \mu = 0.05/252, \gamma_{0} = 0.04/252, \gamma_{1} = 0.5/252, \\
    b_{0} = 0.5, b_{1} =0.372, b_{2} = 0.343, b_{3} = 0.224.
    \end{align*}
The normalized parameter values above imply the daily time unit, and we adapted the estimated coefficients studied in \cite{corsi2009simple} to generate $\tilde{\sigma}_{i}$.
We note that in each simulation, the instantaneous volatility process was generated until all instantaneous volatility values were positive, based on the data-generating process described above.
We set $m =$ 23,400, which indicates that the data are observed every second over a period of 6.5 trading hours per day.

For each simulation, we used the jump robust pre-averaging method \citep{figueroa2022kernel} to estimate the instantaneous volatility, $c_{i,\tau}=\sigma_{i,\frac{\tau}{n}}^2$, at a frequency of every 5 or 10 minutes (i.e., $n=m/300 \text{ or } m/600$) for each $i$-th day as follows: for $\tau = 1, \dots, n$,
\begin{equation}\label{pre-spot}
\hat{c}_{i,\tau} = \frac{1}{\phi_{k_{n}}(g)}\sum_{s=1}^{m-k_{n}+1}K_{b_{m}}(t_{s-1}-\frac{\tau}{n})\left(\bar{Y}_{i,s}^2 - \frac{1}{2}\hat{Y}_{i,s}\right)\mathbf{1}_{\{|\bar{Y}_{i,s}|\leq \nu_{m}\}},
\end{equation}
where $K_{b}(x) = K(x/b)/b$, the bandwidth size $b_{m} = 1/n$, the weight function $g(x) = 2x \wedge (1-x)$, 
\begin{align*}
    &\bar{Y}_{i,s} = \sum_{l=1}^{k_{m}-1}g\left(\frac{l}{k_{m}}\right)(Y_{i,t_{s+l}}-Y_{i,t_{s+l-1}}), \qquad  \phi_{k_{m}}(g) = \sum_{i=1}^{k_{m}}g\left(\frac{i}{k_{m}}\right)^2,\\
    &\hat{Y}_{i,s} = \sum_{l=1}^{k_{m}}\biggl\{g\left(\frac{l}{k_{m}}\right)-g\left(\frac{l-1}{k_{m}}\right)\biggl\}^2 (Y_{i,t_{s+l}}-Y_{i,t_{s+l-1}})^2,
\end{align*}
$\mathbf{1}_{\{\cdot\}}$ is an indicator function, and $\nu_{m} = 1.8\sqrt{\text{BPV}}(k_{m}/m)^{0.47}$, where the bipower variation $\text{BPV} = \frac{\pi}{2}\sum_{s=2}^{m}|Y_{i,t_{s-1}} - Y_{i,t_{s-2}}| |Y_{i,t_{s}} - Y_{i,t_{s-1}}|$.
We used the uniform kernel function and the data-driven approach to obtain the preaveraging window size, $k_m$, as suggested in Section 3.1 of \cite{figueroa2022kernel}.

With the instantaneous volatility estimates spanning $D$ days, $\hat{\bSigma}_{D,n} = (\hat{c}_{i,\tau})_{D\times n}$, we examined the out-of-sample performance of estimating the one-day-ahead instantaneous volatility process.
For comparison, the TIP-PCA, AVE, AR, SARIMA, HAR, HAR-D, XGBoost, and PC methods were employed to predict $c_{D+1,\tau}$, for $\tau = 1, \dots, n$.
Specifically, for the TIP-PCA, we utilized the ex-post daily, weekly, and monthly realized volatilities and the intraday time sequence $\{\frac{\tau}{n}\}_{\tau=1}^{n}$ as covariates for $\bX$ and $\bW$, respectively. 
In addition, the additive polynomial basis with $J_{1}=2$ and $J_{2}=3$ are used for the sieve basis of TIP-PCA, as discussed in Section \ref{tuning}.
AVE represents estimates obtained by the column mean of $\hat{\bSigma}_{D,n}$. 
AR and HAR represent predicted values obtained with the autoregressive model of order 1 and the HAR model, respectively, within each column of $\hat{\bSigma}_{D,n}$.
PC represents the last row of the estimated low-rank matrix using the best rank-$r$ matrix approximation based on $\hat{\bSigma}_{D,n}$.
For TIP-PCA and PC, we used rank 1 as suggested by the eigenvalue ratio method \citep{ahn2013eigenvalue}.
We note that  AR and HAR account for the time series dynamics while AVE and PC cannot.
However, AR and HAR have the overparameterization problem.
PC can partially explain the periodic pattern using the rank-one right singular vector,  while other competitors cannot explicitly account for the pattern due to the random noise $\varepsilon_{i,t_j}$.  
To address intraday periodic patterns in addition to daily autoregressive time series dynamics, we also employed the existing methods below. 
SARIMA represents $n$-step ahead predicted values obtained with the seasonal ARIMA(1,1,1) \citep{sheppard2010financial}, with the length of the seasonal cycle as $n$.
HAR-D represents the modified HAR model that includes the diurnal effect and previous intraday component in addition to the ex-post daily, weekly, and monthly realized volatilities (for details, see \citealp{zhang2024volatility}). 
Lastly, to account for the nonlinear impacts, we included XGBoost \citep{chen2016xgboost}, a decision-tree-based ensemble algorithm. 
We utilized the same hyperparameters as those specified in \cite{zhang2024volatility}.
We note that SARIMA, HAR-D, and XGBoost were conducted after vectorizing $\hat{\bSigma}_{D,n}$ as a time series vector.
We generated high-frequency data with $m =$ 23,400 for 200 consecutive days.
We used the subsampled log prices of the last $D = 50, 100, 150,$ and $200$ days.
To check the performance of the instantaneous volatility, we calculated the mean squared prediction errors (MSPE) as follows:
$$
\frac{1}{n}\sum_{\tau=1}^{n}(\tilde{c}_{D+1,\tau}-c_{D+1,\tau})^{2},
$$
where $\tilde{c}_{D+1,\tau}$ is one of the above one-day-ahead instantaneous volatility estimators. 
Then, we calculated the sample average of MSPEs over 500 simulations.

\begin{figure}
	\includegraphics[width=\linewidth]{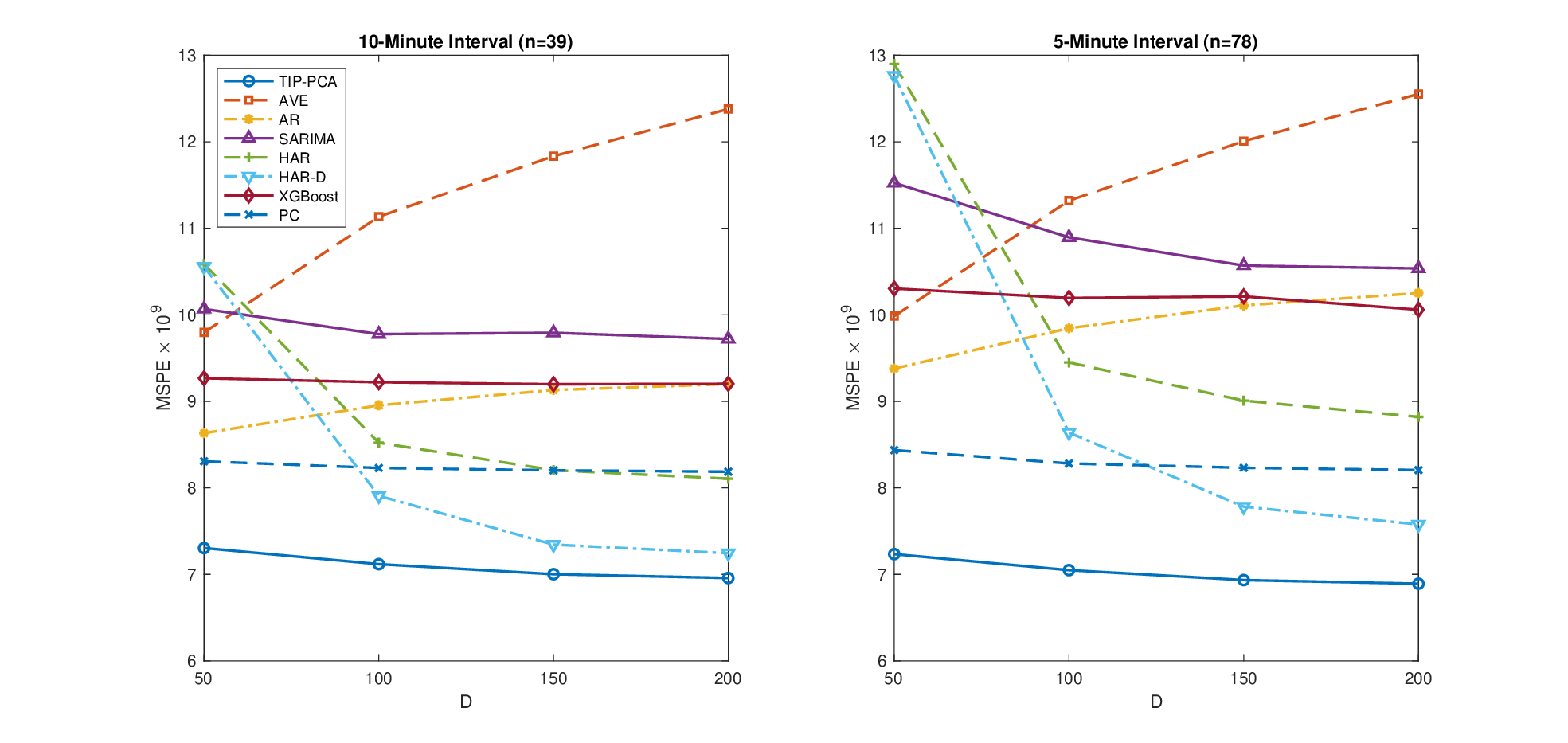}
	\centering	
	\caption{MSPE$\times 10^9$ for the TIP-PCA, AVE, AR, SARIMA, HAR, HAR-D, XGBoost, and PC.}				\label{sim_mspe}
\end{figure}

Figure \ref{sim_mspe} presents the average MSPEs of one-day-ahead intraday instantaneous volatility estimators with $D= \{50, 100, 150, 200\}$ and $n=\{39, 78\}$. 
We note that for each simulation, since we used the subsampled data, the target future volatility is the same for each different $D$.
Figure \ref{sim_mspe} makes evident that the TIP-PCA method demonstrates the best performance. 
This is because TIP-PCA can accurately predict the future instantaneous volatility process by leveraging the features of the HAR model and the U-shaped intraday volatility and handling of the overparameterization problem. 
Additionally, the MSPEs of TIP-PCA tend to decrease as the number of daily observations and intraday instantaneous volatility estimators increases. 
This finding aligns with the theoretical results in Section \ref{asymp}. 
In contrast, the MSPEs of AVE and AR increase as the number of daily observations increases. 
This may be because they do not include the HAR model feature and consider old information deemed unhelpful. 
SARIMA and XGBoost do not perform well because they do not incorporate the interday HAR model feature.
Furthermore, the HAR method does not perform well due to its inability to integrate the U-shaped intraday volatility feature and the overparameterization problem.
By considering the previous intraday component as well as the diurnal effect, HAR-D improves performance compared to HAR. 
However, HAR-D underperforms TIP-PCA because it cannot fully incorporate the intraday periodic pattern.
We note that we also considered a simulation study for the higher rank case ($r=2$) in Section S.1 of the online supplement. 
The result is similar.

\section{Empirical Study} \label{empiric}
In this section, we applied the proposed TIP-PCA method to an intraday instantaneous volatility prediction using real high-frequency trading data.
We obtained intraday data of the S\&P 500 index ETF (SPY) and ETFs that represent the 11 Global Industrial Classification Standard (GICS) sector index funds (XLC, XLY, XLP, XLE, XLF, XLV, XLI, XLB, XLRE, XLK, and XLU) from July 2021 to June 2022 from the TAQ database in the Wharton Research Data Services (WRDS) system.
We collected high-frequency data subsampled every 1 second and excluded days with early stock market closures during this period.
This subsampling helps reduce the effect of irregular observation time points \citep{li2023robust}.
We used the log prices and employed the jump robust pre-averaging estimation procedure defined in Section \ref{simulation} to estimate the instantaneous variance at a frequency of every 10 minutes (i.e., $n=39$).
Then, we conducted the TIP-PCA, AVE, AR, SARIMA, HAR, HAR-D, XGBoost, and PC methods as described in Section \ref{simulation} using the in-sample period data to predict the one-day-ahead instantaneous volatilities. 
In addition, we added TIP-PCA-S, which is a simplified version of the proposed method. 
Specifically, we fit the HAR model based on the row mean of $\hat{\bSigma}_{D,n}$ to predict the one-day-ahead integrated volatility value, while the U-shape is fitted using the column mean of $\hat{\bSigma}_{D,n}$. 
The predicted volatility process is then obtained by multiplying the predicted integrated volatility by the fitted intraday vector.
We note that TIP-PCA-S basically generates a rank one matrix by averaging each dimension of the volatility matrix.
We used the rolling window scheme, where the in-sample period was 63 days (i.e., one quarter).
The out-of-sample forecasting period consisted of $q = 189$ days, resulting in a total of $n \times q = 7{,}371$ predictions (i.e., one prediction every 10 minutes over 189 days).

To measure the performance of the predicted instantaneous volatility, we utilized the mean squared prediction errors (MSPE) and QLIKE \citep{patton2011volatility} as follows:
\begin{align*}
    &\text{MSPE} = \frac{1}{nq}\sum_{i=1}^{q}\sum_{j = 1}^{n}(\tilde{c}_{D+i,j} - \hat{c}_{D+i,j})^{2},\\
    &\text{QLIKE} = \frac{1}{nq}\sum_{i=1}^{q}\sum_{j = 1}^{n}(\log{\tilde{c}_{D+i,j}} + \frac{\hat{c}_{D+i,j}}{\tilde{c}_{D+i,j}}),
\end{align*}
where $\tilde{c}_{D+i,j}$ is one of the TIP-PCA, AVE, AR, SARIMA, HAR, HAR-D, XGBoost, and PC estimates.
We predicted one-day-ahead conditional expected instantaneous volatilities using in-sample period data. 
Additionally, since we do not know the true conditional expected instantaneous volatility, to assess the significance of differences in prediction performances, we conducted the Diebold and Mariano (DM) test \citep{diebold2002comparing} based on MSPE and QLIKE. 
We compared the proposed TIP-PCA method with other methods. 
Importantly, because we conducted multiple hypothesis tests for each ETF, including the DM test based on MSPE and QLIKE, as well as various tests for Value at Risk (VaR), all of which rely on out-of-sample data, it is crucial to control the False Discovery Rate (FDR). 
Therefore, we adjusted the $p$-values using the Benjamini–Hochberg (BH) procedure \citep{benjamini1995controlling} at a significance level of $\alpha = 0.05$ to control the FDR across all hypothesis tests conducted in this section.

\begin{table}[t]
    \centering
    \caption{MSPEs and QLIKEs for the TIP-PCA, AVE, AR, SARIMA, HAR, HAR-D, XGBoost, PC, and TIP-PCA-S.} \label{MSPE and QLIKE}
        \scalebox{0.8}{
    \begin{tabular}{llllllllll}
        \hline
         & TIP-PCA & AVE & AR & SARIMA & HAR & HAR-D & XGBoost & PC & TIP-PCA-S \\
        \midrule
        \multicolumn{10}{c}{MSPE$\times 10^9$} \\
        \cmidrule(lr){2-10}
        SPY & 2.239 & 3.213\textsuperscript{***} & 2.903\textsuperscript{***} & 3.082\textsuperscript{***} & 3.328\textsuperscript{***} & 2.942\textsuperscript{***} & 2.843\textsuperscript{***} & 2.407\textsuperscript{***} & 2.433\textsuperscript{***} \\
        XLC & 3.308 & 4.491\textsuperscript{***} & 4.055\textsuperscript{***} & 4.200\textsuperscript{***} & 4.794\textsuperscript{***} & 4.719\textsuperscript{***} & 4.105\textsuperscript{***} & 3.455\textsuperscript{**} & 3.782\textsuperscript{***} \\
        XLY & 8.606 & 11.992\textsuperscript{***} & 10.327\textsuperscript{***} & 10.375\textsuperscript{***} & 12.162\textsuperscript{***} & 12.399\textsuperscript{***} & 12.533\textsuperscript{***} & 9.157\textsuperscript{***} & 9.507\textsuperscript{***} \\
        XLP & 0.593 & 0.798\textsuperscript{***} & 0.713\textsuperscript{***} & 0.828\textsuperscript{***} & 0.804\textsuperscript{***} & 0.752\textsuperscript{***} & 0.719\textsuperscript{***} & 0.604 & 0.627\textsuperscript{***} \\
        XLE & 8.471 & 11.792\textsuperscript{***} & 9.516\textsuperscript{***} & 9.399\textsuperscript{***} & 10.017\textsuperscript{***} & 10.105\textsuperscript{***} & 10.919\textsuperscript{***} & 8.559 & 8.909\textsuperscript{***} \\
        XLF & 2.380 & 3.510\textsuperscript{***} & 2.854\textsuperscript{***} & 3.121\textsuperscript{***} & 3.085\textsuperscript{***} & 2.897\textsuperscript{***} & 3.118\textsuperscript{***} & 2.524\textsuperscript{***} & 2.507\textsuperscript{***} \\
        XLV & 0.918 & 1.223\textsuperscript{***} & 1.069\textsuperscript{***} & 1.150\textsuperscript{***} & 1.214\textsuperscript{***} & 1.172\textsuperscript{***} & 1.334\textsuperscript{***} & 0.953 & 1.052\textsuperscript{***} \\
        XLI & 1.817 & 2.568\textsuperscript{***} & 2.344\textsuperscript{***} & 2.325\textsuperscript{***} & 2.792\textsuperscript{***} & 2.373\textsuperscript{***} & 2.313\textsuperscript{***} & 1.851 & 1.986\textsuperscript{***} \\
        XLB & 2.182 & 2.993\textsuperscript{***} & 2.605\textsuperscript{***} & 2.819\textsuperscript{***} & 2.965\textsuperscript{***} & 2.891\textsuperscript{***} & 2.912\textsuperscript{***} & 2.254\textsuperscript{**} & 2.376\textsuperscript{***} \\
        XLRE & 1.803 & 2.339\textsuperscript{***} & 2.120\textsuperscript{***} & 2.583\textsuperscript{***} & 2.400\textsuperscript{***} & 2.218\textsuperscript{***} & 2.327\textsuperscript{***} & 1.819 & 1.881\textsuperscript{***} \\
        XLK & 7.477 & 10.233\textsuperscript{***} & 9.019\textsuperscript{***} & 9.157\textsuperscript{***} & 10.628\textsuperscript{***} & 9.819\textsuperscript{***} & 9.538\textsuperscript{***} & 7.879\textsuperscript{***} & 8.159\textsuperscript{***} \\
        XLU & 0.909 & 1.143\textsuperscript{***} & 1.026\textsuperscript{***} & 1.129\textsuperscript{***} & 1.120\textsuperscript{***} & 1.081\textsuperscript{***} & 1.177\textsuperscript{***} & 0.926 & 0.939\textsuperscript{***} \\
        \midrule
        \multicolumn{10}{c}{QLIKE} \\
        \cmidrule(lr){2-10}
        SPY & -9.054 & -8.932\textsuperscript{***} & -9.095\textsuperscript{*} & -6.983\textsuperscript{***} & -8.695\textsuperscript{***} & -7.098\textsuperscript{**} & -9.098\textsuperscript{*} & -9.171\textsuperscript{***} & -9.150\textsuperscript{***} \\
        XLC & -8.871 & -8.801\textsuperscript{**} & -8.898 & -7.061\textsuperscript{***} & -1.790 & -8.378\textsuperscript{***} & -8.897 & -8.973\textsuperscript{***} & -8.970\textsuperscript{***} \\
        XLY & -8.367 & -8.179\textsuperscript{***} & -8.297\textsuperscript{***} & -7.600\textsuperscript{***} & -7.987\textsuperscript{***} & -6.081\textsuperscript{***} & -8.286\textsuperscript{***} & -8.379\textsuperscript{***} & -8.357\textsuperscript{***} \\
        XLP & -9.477 & -9.378\textsuperscript{***} & -9.444\textsuperscript{***} & -7.934\textsuperscript{***} & -9.316\textsuperscript{***} & -8.932\textsuperscript{***} & -9.450\textsuperscript{***} & -9.484\textsuperscript{***} & -9.473\textsuperscript{***} \\
        XLE & -8.090 & -8.013\textsuperscript{***} & -8.060\textsuperscript{***} & -7.993\textsuperscript{***} & -7.947\textsuperscript{***} & -7.971 & -8.068\textsuperscript{***} & -8.093 & -8.085\textsuperscript{***} \\
        XLF & -8.393 & -8.327\textsuperscript{***} & -8.379\textsuperscript{***} & -8.303\textsuperscript{***} & -8.325\textsuperscript{***} & -8.057\textsuperscript{***} & -8.377\textsuperscript{***} & -8.391\textsuperscript{**} & -8.392\textsuperscript{**} \\
        XLV & -9.417 & -9.302\textsuperscript{***} & -9.391\textsuperscript{***} & -6.121\textsuperscript{***} & -9.153\textsuperscript{***} & -6.680 & -9.371\textsuperscript{***} & -9.451\textsuperscript{***} & -9.425\textsuperscript{**} \\
        XLI & -9.134 & -8.997\textsuperscript{***} & -9.109\textsuperscript{**} & -6.366\textsuperscript{***} & -8.579\textsuperscript{***} & -8.638\textsuperscript{***} & -9.122 & -9.197\textsuperscript{***} & -9.358\textsuperscript{***} \\
        XLB & -9.104 & -8.947\textsuperscript{***} & -9.025\textsuperscript{***} & -6.564\textsuperscript{***} & -8.670\textsuperscript{***} & -8.378\textsuperscript{***} & -9.047\textsuperscript{***} & -9.118\textsuperscript{***} & -9.093\textsuperscript{***} \\
        XLRE & -9.082 & -8.986\textsuperscript{***} & -9.036\textsuperscript{***} & -7.208\textsuperscript{***} & -8.845\textsuperscript{***} & -8.244\textsuperscript{**} & -9.029\textsuperscript{***} & -9.075\textsuperscript{*} & -9.067\textsuperscript{***} \\
        XLK & -8.487 & -8.316\textsuperscript{***} & -8.452\textsuperscript{***} & -6.959\textsuperscript{***} & -8.172\textsuperscript{***} & -7.382\textsuperscript{***} & -8.442\textsuperscript{**} & -8.511\textsuperscript{***} & -8.500\textsuperscript{***} \\
        XLU & -9.122 & -9.057\textsuperscript{***} & -9.097\textsuperscript{***} & -7.974\textsuperscript{***} & -9.017\textsuperscript{***} & -8.990\textsuperscript{***} & -9.100\textsuperscript{***} & -9.117\textsuperscript{***} & -9.119\textsuperscript{***} \\

        \bottomrule
    \end{tabular}}
     	\centering
	\begin{tablenotes}
		\item Note: ***, **, and * indicate rejection of the null hypothesis against TIP-PCA at the 1\%, 5\%, and 10\% significance levels, respectively, based on the Diebold-Mariano (DM) test.
	\end{tablenotes}    
\end{table}

Table \ref{MSPE and QLIKE} reports the results of MSPEs and QLIKEs. 
From Table \ref{MSPE and QLIKE}, we find that the TIP-PCA method exhibits the best performance overall. 
This may be  because the projection method, utilizing covariates such as ex-post realized volatility information and the U-shaped intraday volatility feature, contributes to enhancing the accuracy of instantaneous volatility predictions.
We note that PC also performs well compared to other competing methods.
This may be because removing the noise component using the PC method is crucial.  
Additionally, the dynamics of previous intraday volatility significantly influence one-day-ahead intraday volatility forecasts, which indicates a strong AR structure in the interday dimension.
TIP-PCA-S is also comparable to the proposed TIP-PCA in some cases. 
This may be because both TIP-PCA and TIP-PCA-S leverage the same information captured by the HAR model structure within the instantaneous volatility matrix.

We also evaluated the performance of the proposed method in estimating one-day-ahead 10-minute frequency Value at Risk (VaR).
In particular, we first predicted the one-day-ahead conditional expected instantaneous volatilities using the TIP-PCA, AVE, AR, SARIMA, HAR, HAR-D, XGBoost, and PC procedures using the in-sample period data.
We then calculated the quantiles using historical standardized 10-minute returns. 
Specifically, we standardized in-sample 10-minute returns using estimated conditional instantaneous volatilities. 
We then derived sample quantiles for 0.01, 0.02, 0.05, 0.1, and 0.2. 
Using the sample quantile estimates and predicted instantaneous volatility, we obtained the one-day-ahead 10-minute frequency VaR values for each prediction method. 
We used a fixed in-sample period as one quarter and implemented a rolling window scheme.
The out-of-sample period was considered to be from October 2021 to June 2022.

\begin{table}[t]
    \centering
    \caption{Number of cases where the adjusted $p$-value is greater than 0.05 for TIP-PCA, AVE, AR, SARIMA, HAR, HAR-D, XGBoost, PC, and TIP-PCA-S across 12 ETFs at each $q_0 = \{0.01, 0.02, 0.05, 0.1, 0.2\}$ based on the LRuc, LRcc, and DQ tests.}\label{VaR_results}    
    \scalebox{0.87}{
    \begin{tabular}{lrrrrr|rrrrr|rrrrr}
        \toprule
        & \multicolumn{5}{c}{LRuc} & \multicolumn{5}{c}{LRcc} & \multicolumn{5}{c}{DQ} \\
        \cmidrule(lr){2-6} \cmidrule(lr){7-11} \cmidrule(lr){12-16}
        $q_{0}$& 0.01 & 0.02 & 0.05 & 0.1 & 0.2 & 0.01 & 0.02 & 0.05 & 0.1 & 0.2 & 0.01 & 0.02 & 0.05 & 0.1 & 0.2\\
        \midrule
        TIP-PCA	&	9	&	12	&	12	&	12	&	12	&	10	&	12	&	12	&	12	&	12	&	6	&	6	&	7	&	8	&	10	\\
        AVE	&	3	&	3	&	6	&	12	&	12	&	3	&	4	&	5	&	12	&	12	&	3	&	3	&	4	&	4	&	7	\\
        AR	&	4	&	5	&	11	&	12	&	12	&	7	&	7	&	11	&	12	&	12	&	4	&	4	&	5	&	7	&	9	\\
        SARIMA	&	1	&	1	&	1	&	3	&	10	&	7	&	9	&	11	&	12	&	12	&	2	&	2	&	2	&	3	&	6	\\
        HAR	&	1	&	2	&	4	&	12	&	12	&	0	&	1	&	4	&	11	&	12	&	3	&	4	&	4	&	5	&	7	\\
        HAR-D	&	1	&	0	&	4	&	7	&	12	&	7	&	10	&	12	&	12	&	12	&	2	&	2	&	3	&	4	&	6	\\
        XGBoost	&	3	&	4	&	8	&	11	&	12	&	7	&	10	&	11	&	12	&	12	&	3	&	3	&	5	&	6	&	9	\\
        PC	&	6	&	8	&	12	&	12	&	12	&	7	&	9	&	12	&	12	&	12	&	5	&	5	&	7	&	8	&	10	\\
        TIP-PCA-S	&	4	&	6	&	11	&	12	&	12	&	8	&	10	&	12	&	12	&	12	&	4	&	5	&	6	&	8	&	10	\\
        \bottomrule
        \end{tabular}}
	\begin{tablenotes}
		\item Note: The adjusted $p$-values are based on the Benjamini–Hochberg (BH) procedure \citep{benjamini1995controlling} at a level $\alpha = 0.05$.
	\end{tablenotes}        
\end{table}

To backtest the estimated VaR, we conducted the likelihood ratio unconditional coverage (LRuc) test \citep{kupiec1995techniques}, the likelihood ratio conditional coverage (LRcc) test \citep{christoffersen1998evaluating}, and the dynamic quantile (DQ) test with lag 4 \citep{engle2004caviar}.
Table \ref{VaR_results} reports the number of cases where the adjusted $p$-value using the BH procedure is greater than 0.05 for the 12 ETFs at each $q_0 = \{0.01, 0.02, 0.05, 0.1,0.2\}$ quantile, based on the LRuc, LRcc, and DQ tests.
In Table \ref{VaR_results}, we find that the TIP-PCA method consistently outperforms in all hypothesis tests.
We note that both PC and TIP-PCA-S also perform well compared to other methods. 
However, the number of cases for TIP-PCA exceeds those of PC and TIP-PCA-S, particularly for the lower quantiles, which are difficult to predict.
This is because TIP-PCA-S cannot filter out the noise component using the given singular vector structure, and PC cannot explain the complex dynamic structure. 
In contrast, TIP-PCA can filter out the noise component through the projection approach and catch the complex dynamics using the inter-day HAR  and intra-day U-shape structures.
This outcome confirms that the proposed TIP-PCA method significantly contributes to the improved prediction accuracy of future instantaneous volatilities and enhanced risk management.

\section{Conclusion} \label{conclusion}
This paper introduces a novel intraday instantaneous volatility prediction procedure.
The proposed Two-sIde-Projected-PCA (TIP-PCA) method leverages both interday and intraday volatility dynamics based on the semiparametric structure of the low-rank matrix of the instantaneous volatility process. 
We establish the asymptotic properties of TIP-PCA and its future instantaneous volatility estimators.
In the empirical study, concerning the out-of-sample performance of predicting the one-day-ahead instantaneous volatility process, TIP-PCA outperforms other conventional methods. 
This finding confirms that both the HAR model structure on interday dynamics and the U-shaped pattern on intraday dynamics contribute to predicting the future instantaneous volatility process.

In this paper, we focus on the instantaneous volatility process for a single asset. 
In practice, we often need to handle a large number of assets. 
Thus, it is important and interesting to extend the study to predict the instantaneous volatility process of many assets.
However, to do this, cross-sectionally, we encounter another curse of dimensionality. 
Therefore, technically, it is a demanding task to handle both the cross-sectional curse of dimensionality and the intraday curse of dimensionality.
We leave this for a future study.

\bibliography{myReferences}

	\end{document}